\documentclass[twoside,leqno,twocolumn]{article}  
\usepackage{ltexpprt}

\usepackage{latexsym}
\usepackage{epsfig}
\usepackage{cite}
\usepackage{amssymb}
\usepackage{amsmath}
\usepackage{psfrag}
\usepackage{verbatim}
\usepackage{paralist}

\newcommand{\cS}{{\cal S}}

\newcommand{\flo}[1]{\lfloor #1 \rfloor}

\newcommand{\poly}{\textrm{poly}}

\newcommand{\otilde}{\widetilde{O}}
\newcommand{\qed}{\hfill $\Box$}
\newcommand{\eps}{\varepsilon}

\newcommand{\prob}{{\rm Prob}}

\newcommand{\defect}{\mbox{\rm defect}}
\newcommand{\dmin}{\mbox{\rm min-defect}}

\begin{document} 

\title{\Large Space efficient streaming algorithms for the distance to monotonicity 
and asymmetric edit distance}


\author{Michael Saks\thanks{This work was supported in part by NSF under CCF 0832787.}
  \\\\ {\tt saks@math.rutgers.edu} \\
Dept. of Mathematics \\ Rutgers University 
\and C. Seshadhri\thanks{This work was supported by the Early Career LDRD program at Sandia
National Laboratories.}\\\\ {\tt scomand@sandia.gov}\\
Sandia National Labs\thanks{Sandia National Laboratories is a multi-program laboratory managed and operated by Sandia Corporation, a wholly owned subsidiary of Lockheed Martin Corporation, for the U.S. Department of Energy's National Nuclear Security Administration under contract DE-AC04-94AL85000.}\\}

\date{}

\maketitle


\maketitle

\begin{abstract} 
\small\baselineskip=9pt
Approximating the length of the longest increasing sequence (LIS) of an array is a well-studied problem.  We study
this problem in the data stream model, where the algorithm is allowed to make a single left-to-right pass
through the array and the key resource to be minimized is the amount of additional memory used.  
We present an algorithm which, for any $\delta > 0$, given streaming access to an array of length $n$ provides
a $(1+\delta)$-multiplicative approximation to the \emph{distance to monotonicity} ($n$ minus the length of the LIS), and uses only $O((\log^2 n)/\delta)$ space. The previous best known approximation using polylogarithmic space was a multiplicative $2$-factor.   The improved approximation factor reflects a qualitative difference between our
algorithm and previous algorithms: previous polylogarithmic space
algorithms could not reliably detect increasing subsequences
of length as large as $n/2$, while ours can detect increasing subsequences of length $\beta n$ for any
$\beta >0$.  More precisely,  our algorithm can  be used to estimate the length of the LIS to within
an additive $\delta n$ for any $\delta >0$ while previous algorithms could only achieve additive error
$n(1/2-o(1))$. 

Our algorithm is very simple, being just 3 lines of pseudocode, and has a small
update time. It is essentially a polylogarithmic space approximate
implementation of a classic dynamic program that computes the LIS.

We also show how our technique can be applied
to other problems solvable by dynamic programs. 
For example, we give a streaming algorithm
for approximating $LCS(x,y)$, the length of the 
longest common subsequence between strings $x$ and $y$, each of length $n$.  Our
algorithm works 
in the asymmetric setting (inspired by \cite{AKO10}),
in which we have random access to $y$ and streaming access to $x$,
and runs in small space provided that no single symbol appears very often in $y$.
More precisely,
it gives an additive-$\delta n$
approximation to $LCS(x,y)$ (and hence also to $E(x,y) = n-LCS(x,y)$, the edit distance between $x$ and $y$
when insertions and deletions, but not substitutions, are allowed), with space complexity
$O(k(\log^2 n)/\delta)$, where $k$ is the maximum number of times any
one symbol appears in $y$.

We also provide a deterministic 1-pass streaming algorithm that outputs a $(1+\delta)$-multiplicative
approximation for $E(x,y)$ (which is also an additive $\delta n$-approximation), in the asymmetric setting,
and uses $O(\sqrt{(n\log n)/\delta})$ space. 
All these algorithms are obtained by carefully trading space and accuracy
within a standard dynamic program.

\end{abstract}

\section{Introduction} \label{sec:intro} 

Two classic optimization problems concerning subsequences (substrings) of arrays (strings) are the longest increasing subsequence (LIS) and longest common subsequence (LCS) problems. A string of length $n$ over alphabet
$\Sigma$ is represented as a function $x:[n] \rightarrow \Sigma$.
A subsequence of length $k$ is a string $x(i_1)x(i_2)\ldots x(i_k)$, where $1 \leq i_1 < i_2 < \cdots < i_k \leq n$.
In the LIS problem,  the alphabet $\Sigma$ comes equipped with a (total
or partial) order $\triangleleft$, and we look for the longest subsequence
whose terms are in increasing order.  In the LCS problem we are given two strings $x$ and $y$
and look for the longest string which is a subsequence of each of them.
Note that the LIS of $x$ is the LCS of $x$ and its sorted version.

Both of these problems can be solved by dynamic programs. The LIS
can be found on $O(n\log n)$ time \cite{S61, F75, AD99}.
This is known to be optimal, even for (comparison based) algorithms that only
determine the \emph{length} of the LIS \cite{Ram97}.
The LCS problem has a fairly direct $O(n^2)$ algorithm \cite{CLRS},
which can be improved to $O(n^2/\log^2n)$ \cite{MP80,BFC08}. It is a notoriously 
difficult open problem to improve this bound, or prove some matching lower bounds.

%

It is often natural to focus on the complements of the LIS and LCS lengths,
which are related to some notion of distances between strings.
The \emph{distance to monotonicity} of (the length $n$ string) $x$, denoted $DM(x)$
is defined to be $n-LIS(x)$, and is the
the minimum number of values that need to be changed to make
$x$. The (insertion-deletion) edit distance of (two length $n$ strings) $x,y$, 
denoted $E(x,y)$ is defined to be $n-LCS(x,y)$ and is the 
minimum number of insertions and deletions needed to change one string into the other.
(Note that $E(x,y)$ is bounded between $L(x,y)$ and $2L(x,y)$ where
$L(x,y)$ is the Levenshtein distance, where insertions, deletions, and substitutions are allowed.)
Of course the algorithmic problems of exactly computing $LIS(x)$ and $DM(x)$ are equivalent, but approximating
them can be very different. 

%
%
%
%

In recent years, there has been a lot of attention on giving
approximate solutions for LIS and LCS that are much more efficient
that the basic dynamic programming solutions.
Any improved results for LCS would be very interesting, since the best
known quadratic time solution is infeasible for very large strings. These
problems can be studied in a variety of settings - sampling, streaming,
and communication. The streaming setting
has been the focus of many results \cite{GJKK07,SW07,GG07,EJ08}. The model
for the LIS is that we are allowed one (or constant) passes over the input
string $x$, and only have access to sublinear storage. 

The usual formulation of LCS in the streaming model postulates that we have only one-way access to both 
strings $x$ and $y$. We consider an alternative \emph{asymmetric} model in which we have one-way access
to string $x$ (called the \emph{input string}) but random accesss to string $y$ (called the \emph{fixed string}).  This model is more powerful that the standard one,
but it is still far from clear how to obtain space efficient approximations
to $E(x,y)$ in this model.  (This model was inspired by recent work of \cite{AKO10} concerning the time
complexity of approximating edit distance in the random access model. One part of their work introduced
an asymmetric version of the random access model in which one pays only for accesses to
one of the strings, and established time lower bounds for good approximations that hold
even in this more powerful model.)

\vspace{-5pt}
\subsection{Results} \label{sec:results}

Our first result is a streaming algorithm for approximating the distance
to monotonicity. 

\begin{theorem} \label{thm:lis} There is a randomized one-pass streaming algorithm that for any $\delta>0$, takes
as  input an array (of length $n$), makes one-pass through the array, uses space $O(\delta^{-1}\log^2 n)$
and with error probability $n^{-\Omega(1)}$ outputs an estimate to $DM(x)$
that is between $DM(x)$ and $(1+\delta)DM(x)$. 
\end{theorem}

Previously there was a polylogarithmic time algorithm that gave a factor 
$2$-approximation \cite{EJ08}, and an algorithm that gave arbitrarily
good multiplicative approximations to $LIS(x)$ (which is harder
than approximating $DM(x)$) but required
$\Omega(\sqrt{n})$ space \cite{GJKK07}. 

The improvement in the approximation ratio from 2 to $1+\delta$ (for polylogarithmic space algorithms)
is not just ``chipping away'' at a constant, but provides a significant qualitative difference: 
previous polylogarithmic space algorithms might return an estimate of 0 when the LIS length is $n/2$, while
our algorithm can detect increasing subsequences of length a small fraction of $n$.  More precisely,
it is easy to see that if $V$ is an estimate of $DM(x)$ that is between $DM(x)$ and $(1+\delta)DM$ then $n-V$
is within an additive $\frac{\delta}{1+\delta} n$ of $LIS(x)$, and so our algorithm can provide
an estimation interval for $LIS(x)/n$ of arbitrarily small width.
The previous polylogarithmic time streaming algorithm only gave such an algorithm for $\delta \geq 1$, which
only guarantees an estimation interval for $LIS(x)/n$ of width 1/2.

The algorithm promised by
Theorem \ref{thm:lis} is derived as a special case of a more general algorithm (Theorem \ref{thm:main}) that
finds increasing sequences
in partial orders. This algorithm will also be applied to give a good (additive) approximation algorithm for edit distance
in the asymmetric setting, whose space is polylogarithmic  in the case that no symbol appears many times
in the fixed string.
\vspace{-3pt}
\begin{theorem} \label{thm:lcs-main}
Let $\delta \in (0,1]$.
Suppose $y$ is a fixed string of length $n$ and $x$  an input string of length $n$ to which we have streaming access. 
\begin{enumerate}
\item
There is a randomized algorithm that makes one pass through $x$ and, with error probability
$n^{-\Omega(1)}$,  
outputs an additive $\delta n$-approximation to $E(x,y)$ and uses
space $O(k\log^2 n/\delta)$ where $k$ is the maximum number of times any symbol
appears in $y$.
\item 
There is a deterministic algorithm that runs in space 
$O(\sqrt{(n\log n)/\delta})$-space and outputs a $(1+\delta)$-multiplicative (which is also a $\delta n$-additive) approximation to $E(x,y)$.
\end{enumerate} 
\end{theorem}
%

\vspace{-5pt}
\subsection{Techniques} \label{sec:tech}

A notable feature of our algorithm is its conceptual simplicity.  The pseudocode for
the LIS approximation is just
a few lines. The algorithm has parameters $\alpha(i,t)$ for $1 \leq i < t$) whose 
exact formula is
a bit cumbersome to state at this point. We set $\alpha(i,t)$ to be 0 for $i \geq t-O(\log(n))$
and approximately $1/(t-i)$ otherwise.
The algorithm maintains a set of indices $R$, and for each $i \in R$, we store $x(i)$ and
an \emph{estimate} $r(i)$ of $DM(x[1,i])$, where $x[1,i]$ is the length $i$ prefix of $x$.
For convenience, we add dummy elements $x(0) = x(n+1) = -\infty$ and
begin with $R = \{0\}$. For each time $t \geq 1$, we perform the following update:
\vspace{5pt}
\begin{asparaenum}
\item Define $R' = \{i \in R | x(i) \leq x(t)\}$. Set $r(t)=\min_{i \in R'}(r(i)+t-1-i)$.
\item $R \longleftarrow R \cup \{t\}$.
\item Remove each $i \in R$ independently with probability $\alpha(i,t)$.
\end{asparaenum}
\vspace{5pt}
The final output is $r(n+1)$.  

The space used by the algorithm is (essentially) the maximum size of $|R|$.
The update time is determined by step 1, which runs in time $O(|R|)$.
Without the third step, the algorithm is a simple quadratic time exact algorithm
for $DM(x)$ using linear space. (The $O(n\log n)$ time algorithms \cite{F75,AD99} also work in
a streaming fashion, but store data much more cleverly.)
More precisely, at step $t$, $R=\{0,\ldots,t\}$ and 
for each $i \leq t$, $r(i)=DM(x[1,i])$.  

The third step reduces the set $R$, thereby reducing the space of the algorithm.
The space used by the algorithm is (essentially) the maximum size
of $R$.
Intuitively, the algorithm ``forgets''  $(x(i),r(i))$ for those $i$ removed
from $R$. The set $R$ of remembered indices is a subset of $[1,t]$ whose density decays as one goes back in time
from the present time $t$.  When we compute $r(t)$, it may no longer be equal to the distance
to monotonicity of the prefix of $x$ of length $i$, but it will be at least this value.  
This forgetting strategy is tailored to ensure that $r(t)$ is also at most
a $(1+\delta)$-factor away from the distance to monotonicity.
We also ensure that that (with high probability) the set $R$ does not exceed size $O(\delta^{-1}\log^2 n)$  

Just to give an indication of the difficulty, consider
an algorithm that forgets uniformly at random. At some time $t$, the set
of remembered indices $R$ is a uniform random set of size $O(\delta^{-1}\log^2 n)$ up to index $t$.
These are used to compute $r(t)$ and include $t$ in $R$.
The algorithm then forgets a uniform random index in $R$ to maintain the space bound. 
Since we want to get a $(1+\delta)$-factor approximation, the algorithm must be able to detect
an LIS of length $\Omega(\delta n)$. Of the indices in $R$ (up to time $t$),
it is possible that around a $O(\delta)$-fraction of them are in the LIS. 
Suppose we reach a small stretch of indices not on the LIS. If this has size
even $\poly(\delta^{-1}\log n)$, it is likely that all LIS indices in $R$
are forgotten.

But how do we selectively remember the LIS indices without knowing the LIS in
advance? That is the challenge of the forgetting strategy.

%
%
%
%
%
%
%
All past polylogarithmic space algorithms \cite{GJKK07,EJ08} for LIS use combinatorial characterizations
of increasing sequences based on inversion counting \cite{EKK+00,DGLRRS99,PRR04,ACCL1}.
While this is a very powerful technique, it does not lead to accurate approximations
for the LIS, and (apparently) do not yield any generalizations to LCS.

The idea of remembering selected information about the sequence that becomes
sparser as one goes back in time was first used by \cite{GJKK07} for the inversion counting approach.
Our work seems to be the first to use this to directly mimic the dynamic program, though the
idea is quite natural and has almost certainly been considered before.
The main contribution here is to analyze this
algorithm, and determine the values of parameters that allow it to be both
space efficient and a good approximation.

This line of thinking can be exploited to deal with asymmetric streaming LCS. We construct
a simple reduction of LCS to finding the longest chain in a specific partial order. This
reduction has a streaming implementation, so the input
stream can be directly seen as just elements of this resulting partial order.
This reduction blows up the size of the input, and the size of the largest
chain can become extremely small. If each symbol occurs $k$ times in $x$ and $y$,
then the resulting partial order has $nk$ elements. Nonetheless, the longest chain
still has length at most $n$. We require very accurate estimates
for the length of the longest chain. This is where the power of the $(1+\delta)$-approximation
comes in. We can choose $\delta$ to be much smaller to account for the input blow up,
and still get a good approximation. Note that if we only had a $1.01$-approximation
for the longest chain problem, this reduction would not be useful.

Our $\otilde(\sqrt{n})$-space algorithm also works according to the basic principle of 
following a dynamic program, although it uses one different from the previous algorithms. 
This can be thought of as generalization of the $\otilde(\sqrt{n})$-space algorithm for LIS\cite{GJKK07}.
We maintain a $\otilde(\sqrt{n})$-space
deterministic sketch of the data structure maintained by the exact algorithm. 
By breaking the stream up into the right number of chunks, we can update this sketch using 
$\otilde(\sqrt{n})$-space.
\vspace{-5pt}
\subsection{Previous work} \label{sec:prev}
The study of LIS and LCS in the streaming setting was initiated by Liben-Nowell et al \cite{LNVZ05},
although their focus was mostly on exactly computing the LIS. Sun and Woodruff \cite{SW07} improved
upon these algorithms and lower bounds and also proved bounds for the approximate version.
Most relevant for our work, they prove that randomized protocols that compute a $(1+\eps)$-approximation
of the \emph{LIS length} essentially require $\Omega(\eps^{-1}\log n)$. Gopalan et al \cite{GJKK07}
provide the first polylogarithmic space algorithm that approximates the distance to monotonicity. This
was based on inversion counting ideas in \cite{PRR04,ACCL1}. Ergun and Jowhari \cite{EJ08} give a $2$-approximation using the basic
technique of inversion counting, but develop a different algorithm. G\'{a}l and Gopalan \cite{GG07} and
independently Ergun and Jowhari \cite{EJ08} proved an $\Omega(\sqrt{n})$ lower bound for
deterministic protocols that approximate that LIS length up to a multiplicative constant factor. 
For randomized protocols, the Sun and Woodruff bound of $\Omega(\log n)$ is the best known.
One of the major open problems is to get a $o(\sqrt{n})$ space randomized protocol
(or an $\Omega(\sqrt{n})$ lower bound)
for constant factor approximations for the LIS length. Note that our work does not imply
anything non-trivial for this problem. We are unaware of any lower bounds
for estimating the distance to monotonicity in the streaming setting.

A significant amount of work has been done in studying the LIS (or rather, the distance to monotonicity)
in the context of property testing \cite{EKK+00,DGLRRS99,FischerSurvey,PRR04,ACCL1}.
The property of monotonicity has been studied over a variety of domains, of which
the boolean hypercube and the set $[n]$ (which is the LIS setting) have usually been of special
interest~\cite{GGLRS00,DGLRRS99,ELNRR02,HK03,ACCL1,PRR04,BGJRW09}.

In previous work, the authors of this paper found a $(1+\delta)$-multiplicative approximation algorithm for the distance to monotonicity (in the random access model)
that runs in time $O(\poly\log(n))$ \cite{SS10}.  As the present result does for the streaming model,
that result also improved on the previous best factor 2 approximation for that model. 
Despite the superficial similarity between the statement of results, the models considered
in these two papers are quite different, and the algorithm we give here in the streaming model
is completely different from the complicated algorithm we gave in the sublinear time model.

The LCS and edit distance have an extemely long and rich history, especially in the applied domain. We point
the interesting reader out to \cite{Gus97, Nav01} for more details. Andoni et al \cite{AKO10} achieved a 
breakthrough by giving a near-linear time algorithm (in the random access model)
that gives polylogarithmic time approximations
for the edit distance. This followed a long line of results well documented in \cite{AKO10}.
They initiate the study of the \emph{asymmetric edit distance}, where one string is known
and we are only charged for accesses to the other string.
For the case of non-repetitive strings,
there has been a body of work on studying the Ulam
distance between permutations~\cite{AK07,AK08,AIK09,AN10}.

\vspace{-5pt}
\section{Paths in posets}

We begin by defining a streaming problem called the {\em Approximate Minimum-Defect Path} problem (AMDP).
We define it formally below, but intuitively, we look at the stream as a sequence of elements from some poset. Our aim is to estimate
the size of the complement of the longest chain, consistent with the stream ordering. This is more general
than LIS, and we will show how streaming algorithms for LIS and LCS can be obtained
from reductions to AMDP.

\subsection{Weighted $P$-sequences and the approximate minimum-defect path problem}

We use $P$ to denote a 
fixed set endowed with a partial order $\triangleleft$.
The partial order relation is given by an oracle which, given
$u,v \in P$ outputs $u \triangleleft v$ or $\neg(u \triangleleft v)$.
For a natural number $n$ we write $[n]$ for the set $\{1,2,\ldots,n\}$.

A sequence ${\bf \sigma}=(\sigma(1),\ldots,\sigma(n)) \in P$ 
is called
a {\em  $P$-sequence}. The number of terms ${\bf \sigma}$ is called
the length of ${\bf \sigma}$ and is denoted $|{\bf \sigma}|$; we normally use
$n$ to denote $|{\bf \sigma}|$.  A {\em weighted $P$-sequence} 
consists of a $P$-sequence ${\bf \sigma}$ together with
a sequence $(w(1),\ldots,w(n))$ of nonnegative integers; $w(i)$ is called the
{\em weight} of index $i$.  In all our final applications $w(i)$ will
always be $1$. Nonetheless, we solve this slightly more general weighted version.

We have the following additional definitions:
\vspace{5pt}
\begin{asparaitem}
\item For $t \in [n]$,
${\bf \sigma}_{\leq t}$ denotes the sequence $(\sigma_1,\ldots,\sigma_t)$.
Also for $J \subseteq [n]$, $J_{\leq t}$ denotes the set $J \cap \{1,\ldots,t\}$.
\item
For $J \subseteq [n]$, $w(J)=\sum_{j \in J} w(j)$.
\item
The digraph $D=D({\bf \sigma})$ associated to the $P$-sequence ${\bf \sigma}$ has vertex set 
$[n]$ (where $n=|{\bf \sigma}|$) and arc set
$\{i \rightarrow j: i<j \text{ and } \sigma(i) \triangleleft \sigma(j)\}$.
\item A path $\pi$ in $D({\bf \sigma})$ is called a {\em ${\bf \sigma}$-path}.
Such a path is a sequence $1 \leq \pi_1 < \ldots <\pi_k \leq n$ of indices
with $\pi_1 \longrightarrow \cdots \longrightarrow \pi_k$.
We say that $\pi$ ends at $\pi_k$. 
\item The {\em defect of path $\pi$}, $\defect(\pi)$ is defined to be
$w([n]-\pi)$.  
\item $\dmin({\bf \sigma},w)$ is defined to be the minimum of
$\defect(\pi)$ over all ${\bf \sigma}$-paths $\pi$.
\end{asparaitem}
\vspace{5pt}

We now define the {\em Approximate Minimum-defect path} problem (AMDP). The input
is a weighted $P$-sequence $({\bf \sigma},w)$, an approximation parameter $\delta \in (0,1]$,
and an error parameter $\gamma>0$. The output is a number $A$ such that:
%
%
$
\prob[A \in [\dmin({\bf \sigma},w),(1+\delta) \dmin({\bf \sigma},w)]] \geq 1-\gamma.
$
An algorithm for AMDP that has the further guarantee that
$A \geq \dmin({\bf \sigma},w)$ is said to be a {\em one-sided error algorithm}. 
\vspace{-5pt}
\subsection{Streaming algorithms and the main result}

In a one-pass streaming algorithm, the algorithm has one-way access to the input.
For the AMDP, the input consists of the parameters $\delta$ and $\gamma$
together with a sequence of $n$ pairs
$((\sigma(t),w(t)):t \in [n])$.  We think of the input as arriving in a sequence
of discrete time steps, where $\delta,\gamma$ arrive at time step 0 and
for $t \in [n]$, $(\sigma(t),w(t))$ arrives at time step $t$.

The main complexity parameter of interest is the auxiliary memory needed.
For simplicity, we assume that each memory cell can store any one of the following:
a single element of $P$, an index in $[n]$, or an arbitrary sum $w(J)$ of distinct weights.
%
%
Associated to a weighted $P$-sequence $({\bf \sigma},w)$ we define the parameter: 
$\rho=\rho(w)=\sum_i w_i.$ Typically one should think of the weights as bounded by a polynomial in $n$ and so
$\rho=n^{O(1)}$. The main technical theorem about AMDP is the following.
\begin{theorem}
\label{thm:main}
There is a randomized one-pass streaming algorithm for AMDP that
operates with one-sided error and uses space $O(\frac{\ln(n/\gamma)\ln(\rho)}{\delta})$.
\end{theorem}

In particular, if $\rho=n^{O(1)}$ and $\gamma=1/n^{O(1)}$ 
then the space is $O(\frac{(\ln(n))^2}{\delta})$.
\vspace{-5pt}
\section{The algorithm}

Our streaming algorithm can be viewed as a modification of a 
standard dynamic programming algorithm for exact computation of 
$\dmin({\bf \sigma},w)$. 
We first review this dynamic program. 
\vspace{-5pt}
\subsection{Exact computation of $\dmin({\bf \sigma},w)$}

It will be convenient to extend the $P$-sequence by an element
$\sigma(n+1)$ that is greater than all other elements of $P$.
Thus 
all arcs $j \longrightarrow n+1$ for $j \in [n]$ are present.
Set $w(n+1)=0$.
We define
sequences $s(0),\ldots,s(n+1)$ and $W(0),\ldots,W(n+1)$ as follows. We initialize
$s(0) = 0$ and $W(0) = 0$. For $t \in [n+1]$:
\begin{eqnarray*}
W(t) & = & W(t-1)+w(t)\\
s(t) & = & \min (s(i)+W(t-1)-W(i): i < t \\
& & \ \ \ \ \ \ \ \ \ \ \ \ \ \ \ \ \mbox{ such that } \sigma_i \longrightarrow \sigma_t)).
\end{eqnarray*}
Thus $W(t)=w([t])$. 
It is easy to prove by induction that $s(t)$ is equal to
the minimum of $W(t)-w(\pi)$ over all paths $\pi$ whose maximum
element is $\sigma(t)$.  In particular, $\dmin({\bf \sigma},w)=s(n+1)$.

The above recurrence can be implemented by a one-pass streaming algorithm
that uses linear space (to store the values of $s(t)$ and $W(t)$).

\subsection{The polylog space streaming algorithm}

We denote our streaming algorithm by $\Gamma=\Gamma({\bf \sigma},w,\delta,\gamma)$.
Our approximation algorithm is a natural
variant of the exact algorithm.
At  step $t$ the algorithm computes an approximation $r(t)$ to $s(t)$.  The
difference is that rather than storing $r(i)$ and $W(i)$
for all $i$, we store them only for an evolving subset
$R$ of indices, called the {\em active set}
of indices.  The amount
of space used by the algorithm is proportional
to the maximum size of $R$.

We first define the probabilities $p(i,t)$. Similar quantities were defined in \cite{GJKK07}.
\begin{eqnarray*}
& & q(i,t) \\
& = & \min\left\{1,\frac{1+\delta}{\delta}\ln(4t^3/\gamma)\frac{w(i)}{W(t)-W(i-1)}\right\}\\
& & p(i,i) = 1 \ \ \ \ p(i,t) = \frac{q(i,t)}{q(i,t-1)} \text{ for $t>i$},
\end{eqnarray*}

Note that in the typical case that
$\delta = \theta(1)$ and $\gamma = \log(n)^{-\Theta(1)}$, we have $q(i,t)$ is $\Theta(\ln(n)/(t-i))$.

We initialize $R=\{0\}$, $r(0)=0$ and $W(0)=0$. The following update
is performed for each time step $t \in [n+1]$. The final output is just $r(n+1)$.

\medskip

\begin{asparaenum}
\item $W(t)=W(t-1)+w(t)$.
\item $r(t)=\min (r(i)+W(t-1)-W(i): i \in R 
\mbox{ such that } \sigma_i \longrightarrow \sigma_t)$.
\item The index $t$ is inserted in $R$. Each element $i \in R$ is (independently) discarded with probability $1-p(i,t)$.
\end{asparaenum}

\medskip

\begin{theorem}
\label{thm:analysis}
On input $({\bf \sigma},w,\delta,\gamma)$, the algorithm $\Gamma$ satisfies: 
\medskip
\begin{asparaitem}
\item $r(n+1) \geq \dmin({\bf \sigma},w)$.
\item $\prob[r(n+1) > (1+\delta)\dmin({\bf \sigma},w)] \leq \gamma/2$.
\item The probability that $|R|$ ever exceeds $ \frac{2e^2}{\delta}\ln(2\rho)\ln(4n^3/\gamma)$ is at most $\gamma/2$.
\end{asparaitem}
\end{theorem}

The above theorem does not exactly give what was promised
in Theorem \ref{thm:main}. For the algorithm $\Gamma$, there is a small probability that
the set $R$ exceeds the desired space bound while Theorem \ref{thm:main} promises an upper
bound on the space used.  To achieve the guarantee of Theorem \ref{thm:main} we modify
$\Gamma$ to an algorithm $\Gamma'$ which checks whether $R$ ever exceeds the desired space
bound, and if so, switches to a trivial algorithm which only computes the sum of all weights
and outputs that.  This guarantees that we stay within the space bound, and since the probability
of switching to the trivial algorithm is at most $\gamma/2$, the probability
that the output of $\Gamma'$ exceeds $(1+\delta)\dmin({\bf \sigma},w)$ is at most $\gamma$.

We now prove Theorem \ref{thm:analysis}. The first assertion is a direct consequence of the
following proposition whose easy proof (by induction on $t$) is omitted:
\begin{proposition} 
\label{one-sided error}
For all $j \leq n+1$ we have $r(j) \geq s(j)$ and thus
$r(n+1) \geq \dmin({\bf \sigma},w)$.
\end{proposition} 
The second part will be proved in the two subsection. The final assertion
of Theorem \ref{thm:analysis} showing the space bound is deferred to Appendix~\ref{sec:space}.
\vspace{-5pt}
\subsection{Quality of estimate bound of Theorem \ref{thm:analysis}}
We prove the second assertion of Theorem \ref{thm:analysis}, which is the main
technical part of the proof. Let $R_t$ denote the
set $R$ after processing $\sigma(t),w(t)$.
Observe that the definition of $p(i,j)$ implies:
\begin{proposition}
\label{forget prob}
For each $i \leq t \leq n$, $\prob[i \in R_t]=\prod_{j \in [i,t]}p(i,j)=q(i,t)$.
\end{proposition}
We need some additional definitions.
\medskip
\begin{asparaitem}
\item For $I \subseteq [n+1]$, we denote $[n+1]-I$ by $\bar{I}$.
\item
Let $C$ be the index set of some fixed chain having minimum $\defect$, so that
the minimum defect is equal to $w(\bar{C})$.
We assume without loss of generality that $n+1 \in C$.
\item We write $R^t$ for the subset $R$ at the end of step $t$.
Note that $R^t \subseteq [t]$.  We define $F^t=[t]-R^t$.
An index $i \in R^t$ is said to be {\em remembered at time $t$} and
$i \in F^t$ is said to be {\em forgotten by time $t$}.
\item
Index $i \in C$ 
is said to be {\em unsafe at time $t$} if 
every index in $C \cap [i,t] \subseteq F^t$,
i.e., every index of $C \cap [i,t]$ is forgotten
by time $t$.
We write $U^t$ for the
set of indices that are unsafe at time $t$.
\item
An index $i \in C$ is said to be {\em unsafe} if it is unsafe for some time $t>i$
and is {\em safe} otherwise.  We denote the set of unsafe indices by $U$.
%
On any execution, the set $U$ is determined by the
sequence $R^1,\ldots,R^n$.
\end{asparaitem}
\medskip

\begin{lemma}
\label{r(n+1) UB}
On any execution of the algorithm, $r(n+1) \leq w(\bar{C} \cup U)$.
\end{lemma}

\begin{proof}
We prove by induction on $t$ that if $t \in C$
then $r(t) \leq w(\bar{C}_{\leq t-1} \cup U^{t-1})$.
Assume $t \geq 1$ and that the result holds for $j<t$.
We consider two cases.

{\bf Case i.}
$U^{t-1} = C_{\leq t-1}$.  Then $w(\bar{C}_{\leq t-1} \cup U^{t-1})=W(t-1)$.
By definition $r(t) \leq r(0)+W(t-1)-W(0)=W(t-1)$, as required.

{\bf Case ii.}
$U^{t-1} \neq C_{\leq t-1}$. Let $j$ be the maximum index in $C_{\leq t-1}-U^{t-1}$.
Since $j,t \in C$ we must have $\sigma(j) \longrightarrow \sigma(t)$.
Therefore by the definition of $r(t)$ we have:
$r(t) \leq r(j)+W(t-1)-W(j)$.
By the induction hypothesis we have $r(j) \leq w(\bar{C}_{\leq j-1} \cup U^{j-1})$.
Since $j$ is the
largest element of $C_{\leq t-1}-U^{t-1}$ we have:
$\bar{C}_{\leq t-1}\cup U^{t-1}=\bar{C}_{\leq j-1} \cup U^{j-1} \cup [j+1,t-1]$, and so:
\begin{eqnarray*}
r(t) & \leq & r(j)+W(t-1)-W(j) \\
& \leq & w(\bar{C}_{\leq j-1} \cup U^{j-1} \cup [j+1,t-1]) \\
& \leq & w(\bar{C}_{\leq t-1} \cup U^{t-1}) 
\end{eqnarray*}\qed
\end{proof}

By Lemma \ref{r(n+1) UB} the output of the algorithm
is at most $w(\bar{C})+w(U))=\dmin({\bf \sigma},w)+w(U)$.
It now suffices to prove:
\begin{equation}
\label{goal}
\prob[w(U) \geq \delta w(\bar{C})] \leq \gamma/2.
\end{equation}
Call an interval $[i,j]$ {\em dangerous} if $w(C \cap [i,j]) \leq
w([i,j])(\delta/(1+\delta))$.   
In particular $[i,i]$ is dangerous iff $i \not\in C$.
Call an index $i$ dangerous
if it is the left endpoint of some dangerous interval.
Let $D$ be the set of all dangerous indices.

We define
a sequence $I_1,I_2,\ldots,I_\ell$ of disjoint dangerous intervals as follows.
If there is no dangerous interval then the sequence is empty. 
Otherwise:
\begin{itemize}
\item
Let $i_1$ be the smallest index in $D$ and let $I_1$
be the largest interval with left endpoint  
$i_1$.
\item Having chosen $I_1,...,I_j$, if $D$ contains
no index to the right of all of the chosen
intervals then stop.  Otherwise, let $i_{j+1}$ be the 
least index in $D$ to the right of all chosen intervals
and let $I_{j+1}$ be the largest dangerous interval
with left endpoint $i_{j+1}$.
\end{itemize}

It is obvious from the definition that each successive interval lies entirely
to the right of the previously chosen intervals.
Let $B=I_1 \cup \cdots \cup I_\ell$ and let $\bar{B}=[n]-B$.
We now make a series of observations:

\begin{claim}
\label{off chain is dangerous}
$\bar{C} \subseteq D \subseteq B$.
\end{claim}

\begin{claim}
\label{B weight}
$w(B) \leq w(\bar{C})(1+\delta)$.
\end{claim}

\begin{claim}
\label{high probability}
$\prob[U \subseteq B] \geq 1-\gamma/2$.
\end{claim}

By Claims
\ref{off chain is dangerous} and \ref{high probability},
we have $U \cup \bar{C}  \subseteq B$ with probability 
at least $1-\gamma/2$, and so by Claim \ref{B weight},
$w(U \cup \bar{C}) \leq w(\bar{C})(1+\delta)$ with probability
at least $1-\gamma/2$, establishing (\ref{goal}).

Thus it remains to prove the claims.

\medskip
\noindent
{\bf Proof of Claim \ref{off chain is dangerous}}:
If $i \in \bar{C}$  then, as noted earlier, 
$i$ is dangerous so $i \in D$. 

Now suppose $i \in D$.
By the construction of the sequence of intervals,
there is at least one interval $I_1$ and the left endpoint $i_1$ is
at most $i$.
If $i \in I_1 \subseteq B$,  we're done.  So assume $i \not\in I_1$
and so $i$ is to the right of $I_1$.
Let $j$
be the largest index for which $i$ is to the
right of $I_j$.  Then $I_{j+1}$
exists and $i_{j+1} \leq i$.  Since $I_{j+1}$
is not entirely to the right of $i$ we must have $i \in I_{j+1} \subset B$.

\medskip
\noindent
{\bf Proof of Claim \ref{B weight}:}
For each $I_j$ we have $w(I_j \cap C) \leq w(I_j)\delta/(1+\delta)$.
Therefore $w(I_j \cap \bar{C}) \geq w(I_j)/(1+\delta)$ 
and so $(1+\delta)w(I_j \cap \bar{C}) \geq w(I_j)$.
Summing over $I_j$ we get $(1+\delta) w(\bar{C}) \geq w(B)$.

\medskip
\noindent
{\bf Proof of Claim \ref{high probability}:}
We fix $t \in [n]$ and 
$i \in \bar{B} \cap [t]$ and show $\prob[i \in U^t] \leq \frac{\gamma}{4t^3}$.
This is enough to prove the claim since we will then have:
\begin{eqnarray*}
\prob[U \subseteq B] & = & 1-\prob[\bar{B} \cap U \neq \emptyset] \\
& \geq & 1-\sum_{t=1}^n \prob[\bar{B} \cap U^t \neq \emptyset]\\
& \geq & 1-\sum_{t=1}^n \sum_{i \in \bar{B} \cap [t]} \prob[i \in U^t] \\
& \geq & 1-\sum_{t=1}^n \sum_{i \in \bar{B} \cap [t]} \frac{\gamma}{4t^3}\\
&\geq & 1 - \frac{\gamma}{4} \sum_{t=1}^n \frac{1}{t^2} 
\geq 1- \gamma/2.
\end{eqnarray*}
So fix $t$ and $i \in \bar{B} \cap [t]$.  Since $i \not \in B$,
the interval $[i,t]$ is not dangerous, and so $w(C \cap [i,t]) \geq 
w([i,t])\delta/(1+\delta)$, and so 
\begin{equation} \label{[i,t]}
w([i,t]) \leq \frac{1+\delta}{\delta} w(C \cap [i,t]).
\end{equation}

We have $i \in U^t$ only if every index
of $C \cap [i,t]$ is forgotten by time $t$.
For $j \leq t$,
the probability that index $j \in t$ has been forgotten
by time $t$ is $1-q(j,t)$ so $\prob[i \in U^t] = \prod_{j \in C \cap [i,t]}(1- q(j,t))$.
If $q(j,t)=1$ for any of the multiplicands then the product is 0.  Otherwise
for each $j \in C \cap [i,t]$:
\begin{eqnarray*}
q(j,t) & = & \frac{1+\delta}{\delta}\ln(4t^3/\gamma) \frac{w(j)}{(W(t)-W(j-1)} \\
& \geq & \ln(4t^3/\gamma) \frac{1+\delta}{\delta}\frac{w(j)}{w([i,t])}
\geq \ln(4t^3/\gamma) \frac{w(j)}{w(C \cap [i,t])},`
\end{eqnarray*}
where the final inequality uses (\ref{[i,t]}).
Therefore:
\begin{eqnarray*} 
\prob[i \in U(t)] & \leq & \prod_{j \in C \cap [i,t]}(1- q(j,t)) \\
& \leq & \exp(-\sum_{j \in C \cap [i,t]} q(j,t)) \leq \gamma/4t^3,
\end{eqnarray*} 
as required to complete the proof of Claim \ref{high probability}, and of the second assertion of Theorem \ref{thm:analysis}.
\section{Applying AMDP to LIS and LCS} \label{sec:app}

We now show how to apply Theorem~\ref{thm:main} to LIS and LCS. The application
to LIS is quite obvious. We first set some notation about points in the two-dimensional plane.
We will label the axes as $1$ and $2$, and for a point $z$, $z(1)$ (resp. $z(2)$) refers
to the first (resp. second) coordinate of $z$. We use the standard coordinate-wise
partial order on $z$. So $z \triangleleft z'$ iff $z(1) < z'(1)$ and $z(2) < z'(2)$.

\medskip

\begin{proof} (of Theorem~\ref{thm:lis}) The input is a stream $x(1), x(2), \ldots, x(n)$.
Think of the $i$th element of the stream as the point $(i,x(i))$. So the input
is thought of as a sequence of points. Note that the points arrive in increasing order
of first coordinate. Hence, a chain in this poset corresponds exactly to an increasing
sequence (and vice versa). We set $\gamma = n^{O(1)}$ and $\rho = n$ in Theorem~\ref{thm:main}.
\qed
\end{proof}
\medskip

The application to LCS is somewhat more subtle. Again, we think of the input as a set of points
in the two-dimensional plane. But this transformation will lead to a blow up in size, which
we counteract by choosing a small value of $\delta$.

\begin{theorem} \label{thm:lcs-app} Let $x$ and $y$ be two strings of length where each
character occurs at most $k$ times in $y$. Then there is a $O(\delta^{-1}k\log^2n)$-space algorithm 
for the asymmetric setting that outputs an additive
$\delta n$-approximation of $E(x,y)$.
\end{theorem}

\begin{proof} We show how to convert an instance of approximating $E(x,y)$ in the asymmetric model to an instance of AMDP.
Let $P$ be the set of pairs $\{(i,j)|x(i)=y(j)\}$ under the partial order $(i,j) < (i',j')$ if
$i < i'$ and $j < j'$.  It is easy to see that  common subsequences of $x$ and $y$ correspond
to chains in this poset.

Now we associate to the pair of strings $x,y$ the sequence $\sigma$ consisting of points in $P$
listed lexicographically ($(i,j)$ precedes $(i',j')$ is $i < i'$ or if $i=i'$ and $j < j'$.)  Note that $\sigma$
can be constructed online given streaming access to $x$: when $x(i)$ arrives we generate
all pairs with first coordinate $i$ in order by second coordinate.
Again it is easy to check that common subsequences of $x$ and $y$ correspond to $\sigma$-paths
as defined in the AMDP.  Thus the length of the LCS is equal to the size of the
largest $\sigma$-path.  It is not true that $E(x,y)$ is equal to  $min-defect(\sigma)$ (here we omit the weight
function, which we take to be identically 1), because the length of $\sigma$ is in general longer than $n$.
Given full access to $y$, and a streamed $x$.  We have a bound on $|\sigma|$ of $nk$
since each symbol appears at most $k$ times in $x$.

%

We now argue that an additive $\delta n$-approximation for $E(x,y)$ can be obtained
from a $(1+\delta/k)$-approximation for AMDP of $P$. Let the length of the longest
chain in $P$ be $\ell$ and the min-defect be $m$. Let $d$ be a shorthand for $E(x,y)$. 
We have $\ell + m = |P|$ and $\ell + d = n$. The output of AMDP is an estimate $est$
such that $m \leq est \leq (1+\delta/k)m$. We estimate $d$ by $est_d = est + n - |P|$.
We show that $est_d \in [d,d + \delta n]$.

We have $est_d = est + n - |P| \geq m + n - |P| = n-\ell = d$. We can also get an
upper bound.
\begin{eqnarray*}
est_d & = & est + n - |P| \\
& \leq & m + n - |P| + \delta m/k  \\
& = & d + \delta m/k \ \textrm{(since $|P| - m = \ell$ and $d = n - \ell$)} \\
& \leq & d + \delta n \ \textrm{(since $m \leq |P| \leq nk$)}
\end{eqnarray*}
Hence, we use the parameters $\delta/k, \gamma = n^{O(1)}$ for the AMDP instance
created by our reduction. An application of Theorem~\ref{thm:main} completes the proof.
\qed
\end{proof}
\section{Deterministic streaming algorithm for LCS} \label{sec:lcs}
We now discuss a deterministic $\sqrt{n}$-space algorithm for LCS. This
can be used for large alphabets to beat the bound given in Theorem~\ref{thm:lcs-app}.
For any consistent sequence (CS), the
size of the complement is called the \emph{defect}.
For indices $i,j \in [n]$, $x(i,j)$
refers to the substring of $x$ from the $i$th character up to the 
$j$th character. The main theorem is:

\begin{theorem} \label{thm:lcs} Let $\delta > 0$. We have strings $x$ and $y$
with full access to $y$ and streaming access to $x$. There is a deterministic one-pass streaming algorithm that computes a $(1+\delta)$-approximation to $E(x,y)$
that uses $O(\sqrt{(n \ln n)/\delta})$ space. The algorithm performs $O(\sqrt{(\delta n)/\ln n})$
updates, each taking $O(n^2\ln n/\delta)$ time.
\end{theorem}

The following claim is a direct consequence of the standard dynamic
programming algorithm for LCS \cite{CLRS}.

\begin{claim} \label{clm:basic} Suppose we are given two strings $x$
and $y$, with complete access to $y$ and a one-pass stream through $x$.
There is an $O(n)$-space algorithm that guarantees the following:
when we have seen $x(1,i)$, we have the lengths of the LCS between
$x(1,i)$ and $y(1,j)$, for all $j \in [n]$.
\end{claim}

Our aim is to implement (an approximation of) this algorithm in sublinear space. As before, we maintain
a carefully chosen portion of the $O(n)$-space used by the algorithm. In some sense,
we only maintain a small subset of the partial solutions. Although we do not explicitly
present it in this fashion, it may be useful to think of the reduction of Theorem~\ref{thm:lcs-app}.
We convert an LCS into finding the longest chain in a set of points $P$. We construct a set
of \emph{anchor points} in the plane, which may not be in $P$. Our aim is to just 
maintain the longest chain between pairs of anchor points.

Let $\delta > 0$ be some fixed parameter. We set $\bar{n} = \sqrt{(n\ln n)/\delta}$
and $\mu = (\ln n)/\bar{n} = \sqrt{(\delta \ln n)/n}$.
For each $i \in [n/\bar{n}]$, the set $S_i$ of indices is defined as follows.
$$ S_i = \{\flo{i\bar{n} + b(1+\mu)^r} \big| r \geq 0, b \in \{-1, +1\} \} $$ 
For convenience, we treat $\bar{n}$, $n/\bar{n}$, and $(1+\mu)^r$ as 
integers\footnote{Formally, we need to take floors of these quantities. Our analysis
remains identical.}. So we can drop the floors used in the definition of $S_i$.
Note that the $|S_i| = O(\mu^{-1}\ln n) = O(\bar{n})$. We refer to the family
of sets $\{S_1, S_2, \ldots \}$ by $\cS$. This is the set of anchor points that
we discussed earlier. Note that they are placed according to a geometric grid.

\begin{definition} \label{def:proper} A common subsequence of $x$ and $y$ is
consistent with $\cS$ if the following happens. There exists a sequence
of indices $\ell_1 \leq \ell_2 \leq \ldots \ell_m$ such that $\ell_i \in S_i$
and if character $x(k)$ ($k \in [i\bar{n}, (i+1)\bar{n}]$) in the common subsequence 
is matched to $y(k')$, then $k' \in [\ell_i, \ell_{i+1}]$.
\end{definition} 

We have a basic claim about the LCS of two strings (proof deferred to Appendix \ref{sec:lcs-proof}). This gives us a simple bound on the defect that we shall exploit.
Lemma~\ref{lem:proper} makes an important argument. It argues that the
the anchor points $\cS$ were chosen such that an
$\cS$-consistent sequence is ``almost" the LCS.
\begin{claim} \label{clm:defect-gap} Suppose that $x(i_1), x(i_2), \ldots, x(i_r)$ and $y(j_1), y(j_2), \ldots, y(j_r)$ are identical subsequences of $x$ and $y$, respectively.  Let $i \in [n]$ be arbitrary
and let  $i_a$ be the smallest index of the $x$ subsequence such that $i_a \geq i$.
The defect $n-r$ is at least $|j_a - i|$.
\end{claim}
\vspace{-5pt}
\begin{lemma} \label{lem:proper} There exists an $\cS$-consistent common subsequence of $x$
and $y$ whose defect is at most $(1+\delta)E(x,y)$.
\end{lemma}
\vspace{-5pt}
\begin{proof} We start with an LCS $L$ of $x$ and $y$ and ``round" it to be $\cS$-consistent.
Let $L$ be $x(i_1), x(i_2), \ldots, x(i_r)$ and $y(j_1), y(j_2), \ldots, y(j_r)$.
Consider some $p \in [n/\bar{n}]$, and let $i_a$ be the smallest index larger than $p\bar{n}$. Set $\ell_p$ to be the largest index in $S_p$ smaller than $j_a$. We construct
a new common sequence $L'$ by removing certain matches from $L$.
Consider a matched pair $(x(i_b), y(j_b))$ in $L$. 
If $i_b \in [p\bar{n},(p+1)\bar{n}]$ and $j_b \leq \ell_{p+1}$, 
then we add this pair to $L'$. Otherwise, it is not added.
Note that $j_b \geq \ell_p$, simply by construction. The new common sequence $L'$
is $\cS$-consistent.

It now remains to bound the defect of $L'$. Consider a matched pair $(x(i_b), y(j_b)) \in L$
that is not present in $L'$. Let $i_b \in [(p-1)\bar{n},p\bar{n}]$. This means that ${j_b} > \ell_{p}$. Let $i_c$ be the smallest index larger than $p\bar{n}$. 
So $\ell_{p}$ is the largest index in $S_{p}$ smaller than $j_c$. Let
$\ell_{p} = p\bar{n} + (1+\mu)^r $. We
have $j_c - p\bar{n} = [(1+\mu)^r, (1+\mu)^{r+1}]$. Since $j_b \in [\ell_p, j_c]$, the total possible values for $j_b$ is at most $(1+\mu)^{r+1} - (1+\mu)^r$ $= \mu(1+\mu)^r$.
By Claim~\ref{clm:defect-gap}, $E(x,y) \geq j_c - p\bar{n} \geq (1+\mu)^r$. The
number of characters of $x$ with indices in $[(p-1)\bar{n},p\bar{n}]$ that are not
in $L'$ is at most $\mu E(x,y)$. The total number of characters of $L'$ not in $L$
is at most $\mu (n/\bar{n}) E(x,y)$ $\leq \delta E(x,y)$.
\qed
\end{proof}

The final claim shows how we to update the set of partial LCS solutions consistent
with the anchor points.
The proof of this claim and the final proof of the main theorem (that puts everything together)
is given in Appendix \ref{sec:lcs-proof}.

\begin{claim} \label{clm:extend} Suppose we are given the lengths of the largest $\cS$-consistent
common subsequences between $x(1,i\bar{n})$ and $y(1,j)$, for all $j \in S_i$. Also,
suppose we have access to $x(i\bar{n},(i+1)\bar{n})$ and $y$. Then, we can compute
the lengths of the largest $\cS$-consistent common sequences between $x(1,(i+1)\bar{n})$
and $y(1,j)$ (for all $j \in S_{i+1}$) using $\bar{n}$ space.
The total running time is $O(n\bar{n}^2)$.
\end{claim}

\section{Acknowledgements} The second author would like to thank Robi Krauthgamer and David Woodruff
for useful discussions. He is especially grateful to Ely Porat with whom he discussed LCS to LIS reductions.

\bibliographystyle{alpha}
\bibliography{streaming_lis}

\appendix

\section{The space bound of Theorem \ref{thm:analysis}} \label{sec:space}

The following claim shows that the probability
that $|R_t|$ exceeds the space bound is at most $\gamma/2n$. A union bound
over all $t$ proves the third assertion of Theorem \ref{thm:analysis}.

\begin{claim} \label{clm:space} Let $M=\frac{2}{\delta} \ln(4n^3/\gamma)\ln(e\rho)$.
Fix $t \in [n]$. Then $\prob[|R_t| \geq  e^2M] \leq \gamma/2n$.
%
%
\end{claim}

\begin{proof}
For $i \in [t]$ let $Z_i=1$ if $i \in R_t$ and $0$ otherwise.  Then $|R_t|=\sum_{i \leq t} Z_i$.  Let $\mu=\mathbb{E}[|R_t|]$. 
Below we show that $\mu \leq M$.
We need the following tail bound (which is equivalent to the bound of \cite{AlonS}, Theorem A.12):

\begin{proposition}
Let $Z_1,\ldots,Z_m$ be independent 0/1-valued random variables, let $Z=\sum_i Z_i$, and
let $\mu=\mathbb{E}[Z]$.  Then for any $C \geq 0$,
$\prob[Z \geq C] \leq (e\mu/C)^{C}$.
\end{proposition}  

Applying this proposition with $C=e^2M$ gives $\prob[|R_t| \geq e^2M] \leq e^{-C}$ which is at most $\gamma/2n$ (with a lot
of room to spare). It remains to show that $\mu \leq M$.
We have:
\begin{eqnarray*} \mu & = & \sum_{i=1}^t \mathbb{E}[Z_i] = \sum_{i=1}^t q(i,t) \\
& & \leq \frac{2}{\delta}\ln(4n^3/\gamma) \sum_{i=1}^t w(i)/(W(t)-W(i-1). 
\end{eqnarray*}
%
%
%
We note the following fact.
\begin{proposition}
For $r \geq 1$,
$\sum_{i=r}^t w(i)/(W(t)-W(i-1)) \leq \ln(\frac{e(W(t)-W(r-1))}{w(t)})$.
\end{proposition}
\begin{proof} We prove by backwards induction on $r$. For $r=t$,
the left side is $w(t)/(W(t)-W(t-1)) = 1$, the same as the right side.
Assume up to $r \geq 2$, and we shall prove the statement for $r-1$.
We start with a technical statement.
\begin{eqnarray*}
	& & \ln(\frac{W(t) - W(r-2)}{W(t) - W(r-1)}) \\
	& = & \ln(\frac{W(t) - W(r-2)}{(W(t) - W(r-2)) - w(r-1)}) \\
	& = & - \ln(1 - w(r-1)/(W(t) - W(r-2))) \\
	& \geq & w(r-1)/(W(t) - W(r-2))
\end{eqnarray*}
Combining the induction hypothesis with this inequality,
\begin{eqnarray*}
	& & \sum_{i=r-1}^t \frac{w(i)}{W(t)-W(i-1)} \\
	& = & \sum_{i=r}^t \frac{w(i)}{W(t)-W(i-1)} + \frac{w(r-1)}{W(t) - W(r-2)} \\
  & \leq & \ln(\frac{e(W(t)-W(r-1))}{w(t)}) + \ln(\frac{W(t) - W(r-2)}{W(t) - W(r-1)}) \\
  & \leq & \ln(\frac{e(W(t)-W(r-2))}{w(t)})
\end{eqnarray*}
\qed
\end{proof}

Thus $\sum_{i=1}^t w(i)/(W(t)-W(i-1) \leq \ln(eW(t)/w(t)) \leq \ln(e\rho)$, 
and so $\mu \leq M$. This completes the proof.
\qed
\end{proof}

\section{Proofs from Section \ref{sec:lcs}} \label{sec:lcs-proof}

We first prove another claim from which the proof of Claim~\ref{clm:defect-gap} follows.	

\begin{claim} \label{clm:defect} Given a common subsequence 
$x(i_1), x(i_2), \ldots, x(i_r)$ and $y(j_1), y(j_2), \ldots, y(j_r)$, 
the defect is at least $\max_{k \leq r} (|i_k - j_k|)$.
\end{claim}

\begin{proof} (of Claim~\ref{clm:defect}) Assume wlog that $i_k \geq j_k$. Since the $i_k$th character of $x$
is matched to $j_k$th character of $y$, the length of this common
subsequence is at most $LCS(x(1,i_k), y(1,j_k)) + LCS(x(i_k+1,n), y(j_k+1,n))$.
This can be bounded above trivially by $j_k + (n-i_k) = n - (i_k - j_k)$.
Hence the defect is at least $i_k - j_k$. Repeating over all $k$, we complete the proof.
\qed
\end{proof}

\medskip

\begin{proof} (of Claim~\ref{clm:defect-gap}) The defect is at least $|j_a - i_a|$ (by Claim~\ref{clm:defect}) and
is also at least $|i_a - i|$ (by definition of $i_a$). If either $j_a \in [i,i_a]$
or $i \in [j_a,i_a]$, then the defect is certainly at least $|j_a - i|$. Suppose
neither of these are true. Then $j_a > i_a \geq i$. Let us focus on the characters
of $x$ that are not matched. No character of $x$ with index in $[i,i_a)$ is matched.
The characters in $(i_a,n]$ can only be matched to characters of $y$ in $(j_a,n]$
(since $(x(i_a),y(j_a))$ is a match). So the number of characters in $(i_a,n]$
that are \emph{not matched} is at least $(n-i_a) - (n-j_a)$ $=(j_a - i_a)$.
So the number of unmatched characters in $x$ is at least $j_a - i$.
\qed\\
\end{proof}

\begin{proof} (of Claim~\ref{clm:extend}) Consider some $j \in S_{i+1}$, and set $\bar{x} = x(i\bar{n},(i+1)\bar{n})$. We wish to compute the largest $\cS$-consistent
CS between in $x(1,(i+1)\bar{n})$ and $y(1,j)$. Suppose we look at the portion
of this CS in $x(1,i\bar{n})$. This forms a $\cS$-consistent
sequence between $\bar{x}$ and $y(1,j')$, for some $j' \in S_i$.
The remaining portion of the CS is just the LCS between $\bar{x} = x(i\bar{n},(i+1)\bar{n})$
and $y(j',j)$. Hence, given the LCS length of $\bar{x}$ and $y(j',j)$,
for all $j' \in S_i$, we can compute the length of the largest $\cS$-consistent
CS between $x(1,(i+1)\bar{n})$ and $y(1,j)$. This is obtained by just
maximizing over all possible $j'$.

We now apply Claim~\ref{clm:basic}. We have $\bar{x}$ in hand, and stream
in reverse order through $y(1,j)$. Using $O(\bar{n})$ space, we can
compute all the LCS lengths desired. This gives the length of the largest
$\cS$-consistent CS that ends at $y(j)$. This can be done for all $y(j)$, $j \in S_i$.
The total running time is $O(|S_{i+1}|n\bar{n}) = O(n\bar{n}^2)$.
\qed\\
\end{proof}

\begin{proof} (of Theorem~\ref{thm:lcs}) Our streaming algorithm will compute the length of the longest $\cS$-consistent
CS. Consider the index $i\bar{n}$. Suppose we have currently stored the
lengths of the largest $\cS$-consistent CS between $x(1,i\bar{n})$ and
$y(1,j)$, for all $j \in S_i$. This requires space $O(|S_i|) = O(\bar{n}))$.
By Claim~\ref{clm:extend}, we can compute 
the corresponding lengths for $S_{i+1}$ using an additional $O(\bar{n})$ space.
Hence, at the end of the stream, we will have the length (and defect) of the longest
$\cS$-consistent CS. Lemma~\ref{lem:proper} tells us that this defect is a $(1+\delta)$-
approximation to $E(x,y)$. The space bound is $O(\bar{n})$.

The number of updates is $O(n/\bar{n}) = O(\sqrt{(\delta n)/\ln n})$, and the time for each update is $O(n{\bar{n}}^2) = O((n^2\ln n)/\delta)$.
\qed
\end{proof}

\end{document}